\documentclass{article}
\usepackage{graphicx}
\usepackage{amsmath}
\usepackage{amsfonts}
\usepackage{amssymb}
\usepackage{amsthm}
\usepackage{hyperref}
\usepackage[utf8]{inputenc}

\usepackage[numbers,sort&compress,comma]{natbib}

\usepackage{fullpage}
\usepackage{bbm}
\usepackage{graphicx}
\usepackage{upref}
\usepackage{enumerate}
\usepackage{latexsym}
\usepackage{color,graphics}
\usepackage{comment} 
\usepackage{relsize}
\usepackage{url} 
\usepackage{bm} 
\usepackage{marvosym}

\usepackage[section]{algorithm} 
\usepackage{amsmath,amsthm,amssymb,algorithmic,epsfig,graphicx, tikz}

\usepackage{amsfonts}
\usepackage{varioref}

\usepackage{authblk}

\newtheorem{theorem}{Theorem}[section]
\newtheorem{lemma}[theorem]{Lemma}

\newtheorem{definition}[theorem]{Definition}

\newtheorem{assumption}{Assumption}[section]

\newcommand{\R}{\mathbb{R}}

\def\ket#1{\mathinner{|{#1}\rangle}}

\renewcommand{\part}[2]{\frac{\partial #1}{\partial #2}}

\newcommand{\all}[2]{\begin{align}\label{#2} #1\end{align}}
\newcommand{\al}[1]{\begin{align} #1\end{align}}

\newcommand{\en}[1]{\left ( #1 \right )}
\newcommand{\enc}[1]{\left [ #1 \right ]}
\newcommand{\enp}[1]{\left \{ #1 \right \}}

\newcommand{\poly}[1]{O(\mathrm{poly} (#1))}

\newcommand{\polylog}{ \mathrm{polylog}}

\newcommand{\norm}[1]{\lVert#1\rVert}
\newcommand{\abs}[1]{\left\lvert #1\right\rvert}

\newcommand{\thmref}[1]{\hyperref[#1]{{Theorem~\ref*{#1}}}}
\newcommand{\lemref}[1]{\hyperref[#1]{{Lemma ~\ref*{#1}}}}
\newcommand{\remref}[1]{\hyperref[#1]{{Remark~\ref*{#1}}}}
\newcommand{\corref}[1]{\hyperref[#1]{{Corollary~\ref*{#1}}}}
\newcommand{\eqnref}[1]{\hyperref[#1]{{Equation~(\ref*{#1})}}}
\newcommand{\claimref}[1]{\hyperref[#1]{{Claim~\ref*{#1}}}}
\newcommand{\remarkref}[1]{\hyperref[#1]{{Remark~\ref*{#1}}}}
\newcommand{\propref}[1]{\hyperref[#1]{{Proposition~\ref*{#1}}}}
\newcommand{\factref}[1]{\hyperref[#1]{{Fact~\ref*{#1}}}}
\newcommand{\defref}[1]{\hyperref[#1]{{Definition~\ref*{#1}}}}
\newcommand{\exampleref}[1]{\hyperref[#1]{{Example~\ref*{#1}}}}
\newcommand{\hypref}[1]{\hyperref[#1]{{Hypothesis~\ref*{#1}}}}
\newcommand{\secref}[1]{\hyperref[#1]{{Section~\ref*{#1}}}}
\newcommand{\chapref}[1]{\hyperref[#1]{{Chapter~\ref*{#1}}}}
\newcommand{\apref}[1]{\hyperref[#1]{{Appendix~\ref*{#1}}}}

\title{Quantum option pricing via the Karhunen-Lo\`{e}ve expansion}

\date{\today}

\begin{document}

\author[1]{Anupam Prakash$^\dagger$}

\author[2]{Yue Sun$^\dagger$$^*$}

\author[2]{Shouvanik Chakrabarti}
\author[3]{Charlie Che}
\author[1]{\\Aditi Dandapani}
\author[2]{Dylan Herman}
\author[2]{Niraj Kumar}
\author[2]{Shree Hari Sureshbabu}
\author[4]{\\Ben Wood}

\author[5]{Iordanis Kerenidis$^\ddagger$}

\author[2]{Marco Pistoia$^\ddagger$}

\affil[1]{QC Ware, Palo Alto, USA}
\affil[2]{Global Technology Applied Research, JPMorgan Chase, New York, NY}
\affil[5]{QC Ware, Palo Alto, CA and Paris, France}
\affil[3]{Quantitative Research, JPMorgan Chase, New York, NY}
\affil[4]{Quantitative Research, JPMorgan Chase, London, UK}

\clearpage\maketitle
\thispagestyle{empty}

\def\thefootnote{$\dagger$}\footnotetext{These authors contributed equally to this work.}
\def\thefootnote{$\ddagger$}\footnotetext{Principal investigators.}
\def\thefootnote{$*$}\footnotetext{Email: yue.sun@jpmorgan.com}

\begin{abstract} 
We consider the problem of pricing discretely monitored Asian options over $T$ monitoring points where the underlying asset is modeled by a geometric Brownian motion. We provide two quantum algorithms with complexity poly-logarithmic in $T$ and polynomial in $1/\epsilon$, where $\epsilon$ is the additive approximation error. Our algorithms are obtained respectively by using an $O(\log T)$-qubit semi-digital quantum encoding of the Brownian motion that allows for exponentiation of the stochastic process and by analyzing classical Monte Carlo algorithms inspired by the semi-digital encodings. The best quantum algorithm obtained using this approach has complexity $\widetilde{O}(1/\epsilon^{3})$ where the $\widetilde{O}$ suppresses factors poly-logarithmic in $T$ and $1/\epsilon$. The methods proposed in this work generalize to pricing options where the underlying asset price is modeled by a smooth function of a sub-Gaussian process and the payoff is dependent on the weighted time-average of the underlying asset price. 
\end{abstract} 

\section{Introduction}

Option pricing is a fundamental problem in quantitative finance. The objective of option pricing is to determine fair prices for options. Options are financial instruments that provide the possibility to buy or sell an underlying asset at a predetermined price within a specified time horizon. Two techniques are central to option pricing. The first one is the use of stochastic models for the underlying asset price with the goal of accounting for uncertainty in asset prices over time.  The second one is Monte Carlo sampling, which is a powerful computational technique to simulate various future price scenarios. Both these aspects of option pricing, stochastic models and Monte Carlo sampling, have been studied extensively in the quantum setting. 

Quantum accelerated Monte Carlo methods~\cite{M15}, commonly known as quantum Monte Carlo integration (QMCI), are based on amplitude estimation (AE)~\cite{BHMT02}. AE is a fundamental quantum primitive to achieve quadratic speedups for a variety of statistical estimation problems. In QMCI, the quantum algorithm has access to a function $f: X \to \R$ and the goal is to estimate the expected value of $f$ to some additive error $\epsilon$. 
Montanaro \cite{M15} provided an algorithm to estimate the expected value of $f$ for arbitrary distributions on $X$, where the standard deviation of $f$ is bounded by $\sigma$,  using only $\tilde{O}(\sigma/\epsilon)$ calls to the oracle for $f$. This achieves a quadratic speedup over classical algorithms that require $O(\sigma^{2}/\epsilon^{2})$ calls. 

Amplitude estimation requires repeated evaluation of $f$ in the quantum circuit, motivating a series of recent works aiming to reduce the resource requirements of the algorithm. 
Suzuki \emph{et al.}~\cite{S20} used maximum likelihood estimation to eliminate the quantum Fourier transform (QFT) step in the amplitude estimation algorithm. Aaronson and Rall \cite{AR20} gave another QFT-free algorithm with rigorous guarantees and further work by Grinko~\emph{et al.}~\cite{grinko2019iterative} reduced the constant overheads for the QFT-free amplitude estimation. Lastly, Giurgica-Tiron~\emph{et al.}~\cite{GKLPZ20} 
reduced the circuit depth for AE while still retaining a partial speedup by providing AE algorithms with depth $O(1/\epsilon^{1-\beta})$ and $O(1/\epsilon^{1+\beta})$ queries, thus interpolating between classical sampling and the standard AE algorithms.  

Motivated by the provable speedup in query complexity promised by QMCI, quantum algorithms using QMCI for various types of option pricing and risk analysis problems that traditionally rely on classical Monte Carlo methods have been proposed~\cite{Stamatopoulos_2020,unaryoptionpricing2021,Woerner_2019}.
Application-specific resource analysis has also been performed for the pricing of various types of financial derivatives, such as European and Asian options, autocallables and target accrual redemption forwards (TARFs)~\cite{Chakrabarti_2021}.
For a more comprehensive review on quantum algorithms for option pricing, we refer the reader to~\cite{herman2023quantum}.

To use QMCI for pricing an option, we need to be able to simulate the underlying stochastic process in quantum superposition. That is, we need to prepare a quantum state that coherently encodes randomness and then use the randomness to generate sample paths of the stochastic process under consideration. 
While in principle this can always be done in the same amount of time by quantizing the classical operations, in practice this can result in prohibitively large quantum circuits. Montanaro~\cite{M15} noted that quantum walk methods can provide a quadratic speedup in certain cases and used this to show a quantum algorithm for estimating the partition function of the Ising model. Detailed resource estimates for QMCI methods applied to the option pricing problem are provided in~\cite{Chakrabarti_2021}. 

Prior works on option pricing using QMCI~\cite{Chakrabarti_2021,M15,Stamatopoulos_2020,unaryoptionpricing2021,Woerner_2019} have been almost exclusively based on the \emph{digital encoding} of the stochastic processes, in which values of the random variables in the stochastic process are encoded as computational basis states in a quantum register, and their probabilities are encoded in the amplitudes of the corresponding states.
The digital encoding allows for more computational freedom and flexibility, as any classical operation on classical bits can be accordingly performed on qubits. 
This means that classical algorithms for generating sample paths of the stochastic process and any further transformations of the stochastic process, such as evaluating a payoff function on the sample paths, can be easily ported to quantum algorithms using, e.g., quantum logical operations and quantum coherent arithmetic applied on sample paths in superposition.
On the other hand, however, the digital encoding does not a priori offer any quantum advantage over classical simulation of stochastic processes 
as it emulates the classical algorithm for generating the sample paths of the stochastic process.
This means that the only source of quantum speedup is the final amplitude estimation to estimate the mean.

A different approach was proposed recently by Bouland~\emph{et al.}~\cite{BDP23}, using an ``analog'' quantum representation of stochastic processes where the value of the process at time $t$ is stored in the amplitude of the quantum state. This enables an exponentially more efficient encoding of sample paths in terms of the number of time points. Nevertheless, the fact that the stochastic process values are encoded in the amplitude restricts the type of operations that can be efficiently applied to them.
In particular, compared to digitally encoded stochastic processes, it is much more difficult to apply logical and arithmetic operations to the sample paths with the analog encoding.
This means that the types of option payoff functions that can be implemented using the analog encoding are greatly limited. 

\subsection{Our results and techniques}
In this work, we propose a new type of encoding for stochastic processes that combines the compactness of representation of the analog encoding with the ease of computational manipulation of the digital encoding.
We call this new encoding the \emph{semi-digital encoding}.
Specifically, in the semi-digital encoding, time points along a sample path of the stochastic process are stored in superposition, while the values of the sample path are represented by computational basis states in a separate quantum register, making it accessible for computation. 

Using the semi-digital encoding and an analysis of the truncation error of the Karhunen-Lo\`{e}ve (KL) expansion, we show how to efficiently price discretely monitored Asian options where the underlying asset follows a geometric Brownian motion. We propose four different quantum and classical algorithms for pricing the discretely monitored Asian option using KL expansion. These algorithms are summarized in Table~\ref{table1} along with the standard quantum and classical MCI algorithms for comparison. 

Our study shows that using KL expansion and time-domain sub-sampling can reduce the complexity of payoff evaluation when the number of time points $T$ used in the payoff function is sufficiently large. If the stochastic process is a geometric Brownian motion, speedups are achieved in the regime $T \gg 1/\epsilon^{2}$. We demonstrate the quantum and classical algorithms with the Asian option pricing problem on a stochastic process modeled by geometric Brownian motion, but the methods developed generalize to pricing over exponentiated Gaussian process and other smooth functions of sub-Gaussian processes. 
For stochastic processes with fairly general smoothness assumptions, we show that the time-domain sub-sampling approach can be more efficient than the semi-digital encoding based approach.

The main conceptual tool in our work is the relation between semi-digital encodings of stochastic processes and Gaussian state preparation that arises through the KL expansion of Gaussian processes. The KL expansion can be viewed as an expansion of the stochastic process along the principal components. 
It approximates the process by its projection onto the eigenspace corresponding to the largest $L$ eigenvectors of the covariance matrix. This decouples a Gaussian process into independent Gaussian components along the eigenvectors of the covariance matrix.

Several stochastic processes that arise commonly in mathematical finance, including Brownian motion, fractional Brownian motion and the Ornstein-Uhlenbeck process, have a 
KL expansion where the basis change is given by a Fourier transform. As observed in \cite{BDP23}, this opens up the possibility of obtaining a quantum speedup for generating analog encodings of these stochastic processes by utilizing the quantum Fourier transform circuit. In our work, we show that using the semi-digital encoding enables us to provide efficient algorithms for exponentiated Gaussian processes like the geometric Brownian motion. 

Using the semi-digital encodings we provide a quantum Monte Carlo algorithm for pricing discretely monitored Asian options over $T$ points with complexity $\widetilde{O}(1/\epsilon^{4})$, where the $\widetilde{O}$ notation suppresses factors poly-logarithmic in $T$. This is a speedup over classical Monte Carlo algorithms in the regime $T\gg 1/\epsilon^{2}$. The high complexity of the quantum Monte Carlo method is due to the use of two nested amplitude estimations, which arise due to the specific form of the payoff function for the discretely monitored Asian option. 

The semi-digital encoding approach does not use the quantum Fourier transform and therefore quantum-inspired sampling approaches can be derived from it. Inspired by our quantum techniques and in particular by the KL expansion of Gaussian processes, we provide two classical algorithms for pricing discretely monitored Asian options on GBM in Section~\ref{sec:classical}. Our first classical algorithm is a direct dequantization of the quantum algorithm with KL expansion and semi-digital encoding, in which we achieve a complexity of $\widetilde{O}(1/\epsilon^{6})$, that has a quadratic slowdown compared to the quantum algorithm for the two nested Monte Carlo estimators.

Then, we present a second classical algorithm using direct sub-sampling in the time domain. The algorithm relies on an analysis of the KL expansion that shows that if the eigenfunctions for the KL expansion are sufficiently regular, then one can in fact discretize the time domain into $\max(\frac{\sqrt{L_{X}(\epsilon)}} { \epsilon }, \frac{1}{\epsilon^{2}}) $ parts to obtain an $\epsilon$-approximate estimator for the Asian option payoff. Here, $L_{X}(\epsilon)$ is a suitably defined truncation index that depends on the stochastic process $X$, for the case of Brownian motion $L_{X}(\epsilon) = O(1/\epsilon^{2})$. Thus we are able to match the running time of the semi-digital encoding approach with a classical Monte Carlo algorithm for this case.

The quantized version of the time-domain sub-sampling-based classical algorithm uses amplitude estimation to further speed up over the classical algorithm. We provide such a quantum pricing algorithm in Section~\ref{sec:quantized-algorithms}. With a complexity of $\widetilde{O}(1/\epsilon^{3})$, this is the fastest quantum Monte Carlo integration method for Asian option pricing. 
All four quantum and classical algorithms for pricing discretely monitored Asian options are summarized in Table \ref{table1}. 

Although we present our specific algorithms in the context of pricing Asian options on a geometric Brownian motion, our analyses on the KL expansion and time-domain sampling techniques are more general and hence apply to other stochastic processes and payoffs. 
In particular, we show that under mild regularity conditions on the underlying stochastic process, the complexity of the time-domain sub-sampling algorithm and an algorithm based on KL expansion are related when pricing options with a wide range of continuous payoffs. However, the analysis relies on a trade-off between the eigenvalues in the KL expansion and the Lipschitz constants for the eigenfunctions, the overall complexity of the method will therefore depend on the stochastic process being used to model the underlying asset. 

Table~\ref{table1} may suggest that the semi-digital encoding quantum algorithm does not offer any advantages as it can be dequantized and the fastest quantum algorithm for Asian option pricing over a geometric Brownian motion is obtained by quantizing the time-domain sub-sampling-based classical Monte Carlo algorithm. However, we should note that the high complexity of the semi-digital encoding algorithm comes from the two nested Monte Carlo estimations required by the Asian option payoff, and our asymptotic analysis on the convergence of the nested Monte Carlo is quite general and hence can be potentially tightened. 
Furthermore, it may be possible to improve the KL expansion truncation error analysis for the semi-digital encoding algorithm by utilizing the specifics of the payoff function. Finally, the time-domain sub-sampling may also suffer from additional overheads for stochastic processes with non-smooth KL expansions. 

Therefore, the extent of the quantum speedup is problem specific and techniques like the semi-digital encoding may be helpful for other option pricing problems with different payoffs.  As an example, consider stochastic processes where the truncation $L_{X}(\epsilon)= O(1/\epsilon^{\alpha})$ for $\alpha \in (1,2]$ like the fractional Brownian motion with Hurst parameter $H>1/2$~\cite{BDP23}. The semi-digital encoding based quantum algorithm for Asian option pricing on these processes has complexity $\widetilde{O}(L_{X}(\epsilon)/\epsilon^{2})$, improving upon the classical time-domain sub-sampling algorithm with complexity $\widetilde{O}(1/\epsilon^{4})$. However, the quantized time-domain sub-sampling algorithm with complexity $\widetilde{O}(1/\epsilon^{3})$ remains the fastest algorithm for these cases as well.

\begin{table} \label{tab1} 
\begin{center} 
  \begin{tabular}{|c|c|c|} 
   \hline
  &&\\
   \textbf{Algorithm} & \textbf{Qubits used} & \textbf{Gate complexity}  \\
      &&\\
    \hline
      &&\\
      Semi-digital encoding & $O\en{\mathrm{log}(T)\epsilon^{-2}}$ & $O\en{\mathrm{polylog}(T)\epsilon^{-4}}$ \\
    &&\\
    \hline 
     &&\\
      Quantum-inspired sampling & NA &  $O\en{\epsilon^{-6}}$  \\
    &&\\
    \hline
     
     &&\\
      Time-domain sub-sampling  &  NA &  $O\en{\epsilon^{-4}}$ \\
    &&\\
    \hline
     &&\\
     Quantized sub-sampling   & $O\en{\epsilon^{-2}}$ & $O\en{\epsilon^{-3}}$  \\
      &&\\
    \hline
     &&\\
     Standard quantum MCI   & $O(T)$ & $O(T\epsilon^{-1})$  \\
      &&\\
    \hline
     &&\\
     Standard classical MCI   & NA & $O(T\epsilon^{-2})$  \\
      &&\\

    \hline  
   \end{tabular} 
   \caption {Qubit resource requirement and complexity for quantum and quantum-inspired Asian option pricing algorithms with $T$ discrete monitoring points and accuracy $\epsilon$. Factors poly-logarithmic in $1/\epsilon$ have been suppressed, but those that are poly-logarithmic in $T$ have been kept for clarity. The standard algorithms based on quantum and classical Monte Carlo integration (MCI) are also listed for comparison.
   } \label{table1} 
  \end{center} 
   \end{table}

Additionally, variance reduction techniques such as multi-level Monte Carlo (MLMC) may further reduce the complexity of the algorithms discussed herein. We discuss briefly in Section~\ref{sec:mlmc} that MLMC can be applied similarly to both the time-domain sub-sampling-based approach and the KL-expansion-based approach. 

This manuscript is organized as follows. Semi-digital encodings of stochastic processes and the algorithm for generating these encodings for a general class of processes including the Gaussian and exponentiated Gaussian processes are introduced in Section~\ref{sec2}. The quantum algorithm for pricing Asian options and its generalization using the semi-digital encodings is given in Section~\ref{sec3}. Classical algorithms inspired by the semi-digital encoding approach are discussed in Section~\ref{sec:classical}. The quantized version of the time-domain sub-sampling algorithm is given in Section~\ref{sec:quantized-algorithms}.

\section{Semi-digital encodings of stochastic processes}\label{sec2} 
Quantum encodings of stochastic processes in the literature have so far followed two approaches. 
The digital encoding stores the entire sample path of the stochastic process in a quantum register \cite{Chakrabarti_2021}. This approach allows for pricing of different types of path-dependent options, however it does not offer a quantum advantage over classical simulation of stochastic processes in the number of time steps $T$ as it emulates the classical algorithm for generating the encoding of the stochastic process.

Another approach to quantum encoding for stochastic processes using analog encodings where the stochastic process values are stored in the amplitudes of the quantum state was proposed recently~\cite{BDP23}. Analog encodings offer the possibility of quantum speedups using the quantum Fourier transform as several stochastic processes have compact expressions in the Fourier domain. There are two main difficulties with using analog encodings for option pricing. First, encoding the stochastic process values in the amplitude makes them difficult to compute with and restricts the type of options that can be efficiently priced using analog encodings.
In particular, any option with a payoff that requires arithmetic computation on the individual values along a sample path of the stochastic process other than a simple sum or average would be difficult to implement with the analog encoding.
Further, finding quantum algorithms generating analog encodings of exponentiated Gaussian processes like the geometric Brownian motion requires new techniques as the exponentiated processes do not have a compact expression as a Fourier series. 

It is therefore desirable to have a quantum encoding for stochastic processes that both stores the value of the stochastic process in a separate quantum register, making it accessible for computation, while also utilizing the compact representation of the stochastic processes in the Fourier domain. To this effect, 
we propose a notion of semi-digital encodings of stochastic processes that interpolates between the digital and analog encodings. 
\begin{definition} \label{d2} 
A semi-digital encoding for a single trajectory $ (v(1),v(2),\ldots v(T))$ of the stochastic process $S(T)$ is the quantum state, 
\al{  \ \ket{\psi_{v_1,v_2,\ldots v_T}}=\frac{1}{\sqrt{T}}\displaystyle\sum_{t=1}^T\ket{t}\ket{v(t)}. } 
\end{definition}
\noindent The semi-digital encoding allows for the exponentiation of the stochastic process as the value of the process is stored in an auxiliary register that can be exponentiated. 
In order to use semi-digital encodings with quantum Monte Carlo methods, a stronger notion of coherent semi-digital encodings for $S(T)$ is required.

\begin{definition}
The coherent semi-digital representation of the stochastic process $S(T)$ is a superposition over the semi-digital representations of the corresponding trajectories, 
\al{  \ket{S(T)}= \displaystyle\sum_{v_1,v_2,\ldots v_T} \sqrt{p_{v_1,v_2,\ldots v_T}} \ket{\psi_{v_1,v_2,\ldots v_T}} \ket{r_{v_1,v_2,\ldots,v_T}},  } 
where the second register $\ket{r}$ encodes the randomness used to generate the corresponding trajectory. 
\end{definition} 
\noindent Further, the notion of an $\epsilon$-approximate coherent semi-digital encoding is introduced in order to use the truncated Karhunen-Lo\`{e}ve expansion of the stochastic process. 
\begin{definition}\label{d3} 
An $\epsilon$-approximate coherent semi-digital encoding for $S(T)$ is a superposition over $\epsilon$ approximate trajectories, i.e.
\begin{align}\ket{S_{\epsilon}(T)}= \displaystyle\sum_{v_1,v_2,\ldots v_T} \sqrt{p_{v_1,v_2,\ldots v_T}} \ket{\psi'_{v_1,v_2,\ldots v_T}} \ket{r_{v_1,v_2,\ldots,v_T}}\end{align}
 such that \begin{align}\mathbb{E}_{p} [ \norm{ \ket{\psi'_{v_1,v_2,\ldots v_T}} - \ket{\psi_{v_1,v_2,\ldots v_T}}}_{2} ^{2} 
 ] \leq \epsilon^{2}.\end{align} 
\end{definition}

\noindent We will show that $\epsilon$-approximate (coherent) semi-digital for Gaussian processes can be generated with essentially the same resource requirements as required for generating the corresponding analog encodings using coherent qrithmetic operations. The semi-digital encoding provides the additional flexibility that can be used to exponentiate the corresponding processes and to price Asian options on them.

\subsection{Gaussian state preparation}
 We start by showing how to prepare the semi-digital encoding for the Brownian motion and then the geometric Brownian motion follows by exponentiation using coherent arithmetic operations.
 
 The semi-digital encoding of Brownian motion is prepared using the Fourier decomposition of the Brownian motion as a Wiener series. As each term in the Wiener series is an independent Gaussian random variable, the first step is a procedure for preparing quantum states representing Gaussian random variables. 

The Gaussian state on $n$ qubits represents the standard normal distribution on the interval $[-A, A]$ for some fixed constant $A$. The Gaussian state is defined as, 
\begin{align}
\label{eqn:gaussian_state}
    \ket{G} = \frac{1}{ (2\pi)^{1/4}\sqrt{N}}  \sum_{x=-N/2}^{N/2}e^{-(\frac{Ax}{N})^2}\ket{x} + \ket{\perp},
\end{align}
where $N = 2^n$ and $\ket{\perp}$ is an orthogonal state representing the deviation from the ideal Gaussian state due to truncation to a finite number of qubits. The constant $A$ is chosen to be large enough so that the norm of $\ket{\perp}$ is exponentially small and the finite truncation approximates the continuous Gaussian state. 

The Gaussian state preparation problem has attracted a great deal of attention in the recent literature as it is a fundamental problem on the loading of probability distributions into quantum states. 
A recent result, described in Theorem \ref{t1}, on Gaussian state preparation establishes the following bounds on the gate complexity of a quantum circuit needed 
to prepare a standard Gaussian state on $n$ qubits. 

\begin{theorem} \label{t1} 
\cite{BGM22} The Gaussian state with unit variance as defined in \eqref{eqn:gaussian_state} can be prepared to trace distance error $\epsilon$ with a quantum circuit consisting of $O(n( A+ \log(1/\epsilon)))$ gates.
\end{theorem} 
\noindent The goal of \cite{BGM22} was to avoid coherent arithmetic and as a result their Gaussian state preparation algorithm has a dependence on the variance as well. Here, we use coherent arithmetic for option pricing, it is therefore sufficient to generate Gaussian states with unit variance and then rescale using coherent multiplication if required for the stochastic process encoding. In addition, since we are performing integration via QAE, we do not need to perform the additional amplitude amplification step that occurs in \cite{BGM22}, as is indicated by the additional $N^{-1/2}$ factor in Equation~\eqref{eqn:gaussian_state}. This explains the reduced complexity compared to the original source.

\subsection{Karhunen-Lo\`{e}ve expansion for stochastic processes}
\label{sec:kl-expansion}
The relation between the quantum encodings of stochastic processes and Gaussian state preparation arises through the Karhunen-Lo\`{e}ve expansion of 
Gaussian processes. The Karhunen-Lo\`{e}ve expansion can be viewed as an expansion of the stochastic process along its principal components, the process is 
approximated by its projection onto the eigenspace corresponding to the 
largest $L$ eigenvectors of its covariance matrix. 
This decouples a Gaussian process into independent components along the basis of eigenvectors for the covariance matrix. 

Several Gaussian processes that arise commonly in mathematical finance including Brownian motion, fractional Brownian motion and the Ornstein-Uhlenbeck process have the eigenbasis for the covariance matrix by the Fourier transform. As observed in \cite{BDP23}, this opens up the possibility of obtaining a quantum speedup for generating encodings of these stochastic processes.

In order to state the results on the KL-expansion of the Brownian motion, we first define formally the KL-expansion of a stochastic process and then specialize to the case of Brownian motion. Without loss of generality, let the time interval of interest be $[0,1]$. Let $X_{t}$ be a centered, square-integrable stochastic process defined over some probability space $(\Omega, F, \mathbb{P})$ and indexed over $[0,1]$. Define the continuous covariance function to be $k_{X}(s,t) = \mathbb{E}[X_s X_t]$ for all $s,t \in [0,1]$. Then the existence of a KL decomposition is guaranteed by the Kosambi-Karhunen-Lo\`{e}ve theorem as stated below.
\begin{theorem}
\label{thm:kkl}
    (Kosambi-Karhunen-Lo\`{e}ve~\cite[Theorem~7.21]{shynk2012probability})
    Let $X_t$ be defined as above. Define the linear operator $\mathcal{K}$ on $\mathcal{L}^2([0,1])$ that maps $f \to \mathcal{K}f = \int_{0}^{1} k_{X}(s,\cdot) f(s) \,ds$ and let $\{e_k\}$ be the orthonormal basis on $\mathcal{L}^2([0,1])$ formed by the eigenfunctions of $\mathcal{K}$ with respective eigenvalues $\{\lambda_k\}$. Then $X_t$  admits a representation (called the KL-expansion) that is $t$-uniformly convergent in $\mathcal{L}^2(\mathbb{P})$ norm and given by
    \begin{align}
        X_t = \sum_{k=1}^\infty Z_k e_k(t) \text{ where } Z_k = \int_{0}^{1} X_t e_k(t) dt.
    \end{align}
    The random variables $Z_k$ have zero mean ($\mathbb{E}[Z_k] = 0, \forall k \in \mathbb{N}$) and are uncorrelated with variance $\lambda_k$ ($\mathbb{E}[Z_i Z_j] = \delta_{ij}\lambda_{j}, \forall i,j \in \mathbb{N}$).
\end{theorem}
\noindent The KL-expansion allows us to define a truncation index 
$L_{X}(\epsilon)$ for the stochastic process. 

\begin{definition} \label{def:trunc-index}
Let $X_t$ be a stochastic process with KL expansion given by Theorem~\ref{thm:kkl}, the \textbf{truncation index at error $\epsilon$}, denoted by $L_X(\epsilon)$, is defined to be to be the smallest natural number such that
\begin{align}
    \mathbb{E}\left[\left(\sum_{k=0}^{L_X(\epsilon)} Z_k e_k(t) - X_t\right)^2\right] \le \epsilon^2, \quad \forall t \in [0,1].
\end{align}
\end{definition}
\noindent The $t$-uniform $\mathcal{L}^2(\mathbb{P})$ convergence, described above, of the KL decomposition of $X_t$ ensures that the truncation index $L_X(\epsilon)$ is finite for all $\epsilon > 0$. Specifically, the eigenvalues and eigenfunctions for the Brownian motion are given in Theorem~\ref{thm:specBM}.
\begin{theorem} \label{thm:specBM} 
For the Brownian motion $B(t)$, the Mercer kernel is given by $k_B(s, t)= \mathbb{E}[B(s)B(t)]= \min(s, t)$. The eigenvalues are given as $\lambda_{k} = \frac{1}{ (k-1/2)^{2} \pi^{2} } $ and the corresponding eigenfunctions are given as 
$e_{k}(t)= \sqrt{2} \sin( ((k-1)/2) \pi t)$.
\end{theorem} 
\noindent The Mercer kernel for the Brownian motion can be used to show that the truncation index $L_{X}(\epsilon)$ is 
$O(1/\epsilon^{2})$. 
\begin{theorem} \label{thm:LXBrownian}
The truncation index $L_{X}(\epsilon)= O(1/\epsilon^{2})$ if $X$ is a Brownian motion. 
\end{theorem} 
\begin{proof}
Let $X_{t}$ be a Brownian motion, using the spectrum of the Mercer kernel the $Z_{k}$ are evaluated to be 
$ \sqrt{2} \frac{ a_{k} \sin( (k-1)\pi t/2)} {(k-1/2)\pi}$ where $a_{k}$ are i.i.d. standard normal random variables. The sum $\sum_{k=0}^{L} Z_k e_k(t) - X_t = \sum_{k>L} \lambda_{k} a_{k}  \frac{\sqrt{2} \sin( ((k-1)/2) \pi t)}{ (k-1/2)\pi}$ is a mean $0$ Gaussian random variable with variance upper bounded by $\sum_{k>L} \frac{1}{(k-1/2)^{2}\pi}= O(1/L)$. For $L=O(1/\epsilon^{2})$, the expected squared norm is $O(\epsilon^{2})$ implying that $L_{X}(\epsilon)=O(1/\epsilon^{2})$. 
\end{proof}

The truncation index will be useful for the analysis of quantum inspired algorithms in section \ref{sec:classical}. In order to 
generate semi-digital encodings for the Brownian motion, an equivalent description of the KL expansion given by the Wiener series, as stated in Theorem~\ref{thm:wiener} is used.

\begin{theorem} \cite{W23} \label{thm:wiener} 
For independent random variables $a_{k} \sim \mathcal{N}(0,1)$, the following Fourier series represents Brownian motion on the interval $t \in [0, 1]$, 
\al{ 
\label{eq:wiener-bm}
B(t) = a_{0} t + \frac{\sqrt{2}} {\pi}  \sum_{k = 1}^{\infty} \frac{ a_{k} }{k} \sin(k \pi t).
} 
\end{theorem} 
The truncated Wiener series up to $L$ terms is equivalent to  a degree $L$ polynomial in $\cos(\pi t)$ with coefficients known functions of the $a_{k}$. This follows from the definition of the Chebychev polynomials of the second kind $U_{n}(\cos \theta) \sin(\theta) = \sin( (n+1) \theta )$.
\all{ 
B_{L}(t) =  p_{L}(\bm{a}, \cos(t)) = a_{0} t + \frac{\sqrt{2}}{ \pi} \sin(\pi t) \sum_{k=0}^{L-1} \frac{a_{k}}{k} U_{k} (\cos(\pi t)).
} {eq:wiener-cheb} 
where $a_{k}$ are independent standard $\mathcal{N}(0,1)$
random variables. The exponential of equation \eqref{eq:wiener-cheb} corresponds to geometric Brownian motion with parameters $\mu=0$ and $\sigma=1$. The GBM with parameters $\mu$ and 
$\sigma$ is defined as $\exp(\sigma B_{L}(t) + (\mu-\sigma^{2}/2)t)$, it is also a degree $L$ polynomial in $\cos(\pi t)$ and is obtained by multiplying the coefficients of the polynomial in equation \eqref{eq:wiener-cheb} by $\sigma$ and adding the linear drift term $(\mu-\sigma^{2}/2)t)$.

\subsection{Mapped Karhunen-Lo\`{e}ve expansion for stochastic processes}
\label{sec:mapped-kl-expansion}

Theorem~\ref{thm:kkl} suggests that any zero-mean square-integrable stochastic process can be represented using a KL series. 
Despite the general applicability of the theorem, finding the basis functions for the KL expansion of a general stochastic process is often difficult in practice.
For example, the geometric Brownian motion, one of the most widely used stochastic processes in quantitative finance, does not have a closed-form KL expansion.
In particular, the GBM has non-linear dependence on its past values, making it difficult to find eigenfunctions and eigenvalues for its covariance function.
On the other hand, the GBM may be generated by exponentiating a Brownian motion, which, as shown in Section~\ref{sec:kl-expansion}, has a closed-form KL expansion.
This suggests us to represent a more general stochastic process by mapping a simpler stochastic process which has a closed-form KL expansion.

In this section, we assess the feasibility of approximating a stochastic process by applying a function to the truncated KL expansion of another, and analyze the error due to the KL-expansion truncation.
Specifically, we assume that the process $X_t$ that has a closed-form KL expansion is a sub-Gaussian process -- a natural generalization of the Gaussian processes considered in Section~\ref{sec:kl-expansion}.
We also assume that the target process $S_t$ is linked to $X_t$ by a function $g$ satisfying Assumption~\ref{asm:sub-exp-lipschitz-func}.

\begin{assumption}
\label{asm:sub-exp-lipschitz-func}
There exist constants $K_g$ and $c_g$ such that $g: \Omega \to \mathbb{R}$ satisfies
\al{
g(x) - g(y) \le \max \en { \abs{\exp \en{c_g x} - \exp \en{c_g y}}, K_g \abs{x - y}} \quad \forall x, y \in \Omega.
}
\end{assumption}
To bound the error of approximating $S_t = g(X_t)$ using the truncated KL series of $X_t$, we first prove the following lemma.
\begin{lemma}
\label{lem:mapped-proc-bound}
Suppose $X \in \Omega_X$ and $Z \in \Omega_Z$ are two independent sub-Gaussian random variables, of which $X$ has finite mean $\mu$ and variance $\sigma^2$, and $Z$ has mean zero and variance $\epsilon^2$.
Define two additional dependent random variables $S_X = g(X)$ and $S_Y = g(Y)$, where $Y = X + Z \in \Omega_Y := \enp{x+z : x \in \Omega_X, z \in \Omega_Z}$ and $g: \Omega_X \cup \Omega_Y \to \mathbb{R}$ is a function satisfying Assumption~\ref{asm:sub-exp-lipschitz-func}.
Then we can bound the $\mathcal{L}_2$ distance between $S_X$ and $S_Y$ by
\al{
\mathbb{E} \enc{\en{S_X - S_Y}^2} \le C_1 \epsilon^2 + O(\epsilon^4),
}
where
\al{
C_1 = c_g^2 \exp \en{c_g^2 \sigma^2 + 2 c_g \mu} + K_g.
}
\end{lemma}
\begin{proof}
Using Assumption~\ref{asm:sub-exp-lipschitz-func} on $g$, we have
\al{
\mathbb{E} \enc{\en{S_X - S_Y}^2}
&\le \mathbb{E} \enc{\max \enp{ \en{\exp (c_g X) - \exp (c_g Y)}^2, K_g \en{X - Y}^2}} \\
&\le \mathbb{E} \enc{\en{\exp (c_g X) - \exp (c_g Y)}^2} + K_g \mathbb{E} \enc{ \en{X - Y}^2} \\
&= \mathbb{E} \enc{\en{\exp (c_g X) - \exp (c_g Y)}^2} + K_g \epsilon^2.
}
Since $X$ and $Z$ are independent we have
\al{
\mathbb{E} \enc{\en{\exp (c_g X) - \exp (c_g Y)}^2}
&= \mathbb{E} \enc{\exp (2 c_g X)} \mathbb{E} \enc{ \en{1 - \exp (c_g Z)}^2}.
}
Since $X$ and $Z$ have sub-Gaussian distributions, we have 
\al{
\mathbb{E} \enc{\exp (2 c_g X)} \le \exp \en{c_g^2 \sigma^2 + 2 c_g \mu},
}
and 
\al{
\mathbb{E} \enc{ \en{1 - \exp (c_g Z)}^2}
&= \mathbb{E} \enc{ 1 + \exp (2 c_g Z) - 2 \exp (c_g Z)} \\
& \leq \mathbb{E} \enc{ 1 + \exp (2 c_g Z) - 2(1 + c_g Z)}\\
& = -1 + \mathbb{E} \enc{\exp (2 c_g Z)} - 2 c_g \mathbb{E} \enc{Z}\\
&\leq -1 + \exp \en{c_g^2 \epsilon^2} \\
&= c_g^2 \epsilon^2 + O(\epsilon^4).
}
Therefore,
\al{
\mathbb{E} \enc{\en{S_X - S_Y}^2} 
&\le \exp \en{c_g^2 \sigma^2 + 2 c_g \mu} \en{\exp \en{c_g^2 \epsilon^2} - 1} + K_g \epsilon^2 \\
&= \enc{c_g^2 \exp \en{c_g^2 \sigma^2 + 2 c_g \mu} + K_g} \epsilon^2 + O(\epsilon^4) \\
&= C_1 \epsilon^2 + O(\epsilon^4),
}
which completes the proof.
\end{proof}
Armed with the above lemma, we state below our theorem along with the assumption on the target stochastic process $S_t$ that we are trying to approximate.
To make our analysis more general, we allow the function $g$ to be time-dependent and denote the function at each time $t$ by $g_t$.
\begin{assumption}
\label{asm:mapped-kl-process}
The stochastic process $S_t$ is generated by $S_t = g_t(X_t)$, where $g_t$ is a function satisfying Assumption~\ref{asm:sub-exp-lipschitz-func} at each $t$, and $X_t$ is a stochastic process that has a KL expansion with truncation index $L_X(\epsilon)$ and the random variables $Z_k$ in the KL expansion having sub-Gaussian distributions.
\end{assumption}

\begin{theorem}
\label{thm:mapped-kl-truncation-error}
Suppose $S_t = g_t(X_t)$ is a stochastic process with $S_t$, $X_t$ and $g_t$ satisfying Assumption~\ref{asm:mapped-kl-process} on $t \in [0, 1]$, then $S'_t = g_t(X'_t)$ is an $\epsilon$-approximation of $S_t$ satisfying
\al{
\mathbb{E} \enc{\en{S_t - S'_t}^2} \le C_2 \epsilon^2 \quad \forall t \in [0, 1],
}
where $X'_t$ is the stochastic process constructed by truncating the KL expansion for $X_t$ to $L_X(\epsilon)$ terms and $C_2$ is a constant independent of $t$ and $\epsilon$ for all $\epsilon \to 0$.
\end{theorem}
\begin{proof}
By construction, $X_t$ may be written as $X_t = X'_t + X''_t$.
According to Assumption~\ref{asm:mapped-kl-process}, all random variables in the KL expansion for $X_t$ are sub-Gaussian with mean zero. 
Therefore, $X'_t$ and $X''_t$ are both sub-Gaussian with mean zero.
Applying Lemma~\ref{lem:mapped-proc-bound} to $S_t$ and $S'_t$ for each $t$ we immediately have
\al{
\mathbb{E} \enc{\en{S_t - S'_t}^2} = O(\epsilon^2) \quad \forall t \in [0, 1].
}
Since $t$ is bounded to $[0, 1]$, for $\epsilon \to 0$, there exists a uniform constant $C_2$ such that
\al{
\mathbb{E} \enc{\en{S_t - S'_t}^2} \le C_2 \epsilon^2 \quad \forall t \in [0, 1],
}
which completes the proof. 
\end{proof}

In particular, Theorem~\ref{thm:mapped-kl-truncation-error} applies to the case where a GBM is approximated by exponentiating the truncated Wiener series plus a linear shift $\mu t$.
In this case, $X_t$ is the Brownian motion and $g_t(x) = \exp(x + \mu t)$.
Since the truncation index $L_X(\epsilon) = O(1/\epsilon^2)$ for the Wiener series (Thoerem~\ref{thm:LXBrownian}), the GBM may be approximated up to a mean squared error of $\epsilon^2$ using $L = O(1/\epsilon^2)$ terms in the truncated Wiener series.

\subsection{Semi-digital encoding for geometric Brownian motion}
\label{sec:state-prep-gbm}

We primarily consider option pricing where the underlying asset is modeled by a geometric Brownian motion (GBM), since the GBM is one of the most commonly used stochastic processes for modeling financial asset prices as it corresponds to the solutions of the Black Scholes equation.

Geometric Brownian motion is the Doleans exponential of Brownian motion with drift. In other words, it is the unique solution to the Stochastic Differential Equation (SDE) 
\al{ 
dS_t=\sigma S_tdB_t+\mu S_t dt.
} 
The SDE has two independent parameters $\mu$ and $\sigma$ corresponding to the drift and volatility respectively. The solution to this SDE is given in closed form as, 
\al{ 
G_{\mu, \sigma}(t) =S_0 e^{\sigma B_t+ (\mu-\frac{1}{2}\sigma^{2})t}.
} 
The solution to the SDE is always positive and can thus be used to model asset prices. 

The semi-digital encoding for the GBM is generated by first generating a semi-digital encoding for the Brownian motion and then using coherent arithmetic for exponentiation. The semi-digital encoding for the Brownian motion is in turn prepared by decomposing the stochastic process using KL expansion, and coherently evaluating the eigenfunctions in the KL expansion up to the $L$-th largest eigenvalue as a polynomial using Equation~\eqref{eq:wiener-cheb}. The quantum algorithm for preparing the semi-digital encoding of the GBM is presented as Algorithm~\ref{alg:sdGBM}. 

\begin{theorem} 
Algorithm~\ref{alg:sdGBM} generates an $\epsilon$-approximate semi-digital encoding for geometric Brownian motion with parameters $(\mu, \sigma)$ with cost $\widetilde{O}(L)$ where $L=O(1/\epsilon^2)$ is the number of terms retained in the Wiener series and $\widetilde{O}$ suppresses factors logarithmic in $T, L$
and $1/\epsilon$. 
\end{theorem} 
\begin{proof} 
The correctness of the algorithm, that it generates an $\epsilon$ approximate coherent semi-digital encoding for the geometric Brownian motion (Definition \ref{d3}) as retaining $L=O(1/\epsilon^2)$ terms in the Wiener series and exponentiating approximates GBM to $\mathcal{L}_{2}$ error $\epsilon$ by Theorem 
\ref{thm:mapped-kl-truncation-error}. 

We analyze the arithmetic complexity for each step of Algorithm~\ref{alg:sdGBM}. The first step generates $O(L)$ Gaussian states, the gate complexity for this step is $\widetilde{O}(L)$ where factors logarithmic in $L$ and $1/\epsilon$ are suppressed by Theorem \ref{t1}. 
The coefficients of the degree $L$ polynomial in $\cos(\pi t)$ specified in Equation~\eqref{eq:wiener-cheb} can be pre-computed for step 3 of the algorithm. The arithmetic complexity of evaluating a degree $L$ polynomial is $O(L)$ as exactly $L$ additions and $L$ multiplications are needed for polynomial evaluation using Horner's method~\cite{horner1819xxi}.
The values $\sin(\pi t)$ and $\cos(\pi t)$ can be computed coherently using the time register as the input, and then the semi-digital encoding of the Wiener series is obtained by coherently evaluating the degree $L$ polynomials of $\cos(\pi t)$ 
and $\sin(\pi t)$. The additional overhead for coherent arithmetic is $O(p)$ for addition and $O(p^{2})$ for multiplication if $p$ bits of precision are used.

The arithmetic complexity for the computation of $G_{L, \mu, \sigma}(\bm{a}, t)$ is independent of the values of $\mu$ and $\sigma$ as changing these values rescales the coefficients of the polynomial in Equation~\eqref{eq:wiener-cheb}. 
It follows that the total cost for generating a semi-digital encoding of the geometric Brownian motion using Algorithm~\ref{alg:sdGBM} is $\widetilde{O}(L)$. 

\end{proof} 

\noindent For option pricing, the $\mu$ and $\sigma$ are taken to be some fixed known values and the abbreviated notation $G_{L}(\bm{a}, t)$ is used instead of $G_{L, \mu, \sigma}(\bm{a}, t)$
in the following sections for encodings of the geometric Brownian motion. 

\begin{algorithm} [H]
\begin{algorithmic}[1]

\REQUIRE Parameters $(S_0, \mu, \sigma)$ for geometric Brownian motion $S_t$ on $t \in [0, 1]$. Accuracy $\epsilon$ for semi-digital encoding and truncation level $L=O(1/\epsilon)$ for the Brownian motion Wiener series. 
\ENSURE  A $\epsilon$-approximate semi-digital encoding for geometric Brownian motion $\ket{G_{L}(\mu, \sigma)}$.  

\STATE Prepare $L+1$ copies of the unit variance Gaussian state $|G\rangle$ (Theorem \eqref{t1}) to get the quantum state
\al{ 
\ket{G}^{\otimes (L+1)} = \bigotimes_{k=0}^{L} \left( \sum_{a_k=-N/2}^{N/2} \sqrt{ p_k(a_k) }\ket{a_k} \right) 
=  \sum_{\bm{a}} \sqrt{ p(\bm{a}) } \ket{\bm{a}}.
}
\STATE Append a time register to the quantum state containing equal superposition of all $T$ time steps in $t \in [0, 1]$ for the discrete monitoring
\al{ 
\sum_{\bm{a}} \sqrt{ p(\bm{a}) } \ket{\bm{a}} \otimes \frac{1}{\sqrt{T} } \sum_{t=1}^{T} \ket{t}.
}
\STATE Evaluate the polynomial $p_{L}(\bm{a'}, \cos(t))$ in Equation~\eqref{eq:wiener-cheb} coherently for $a_k' = 2A a_{k}/N$ using $O(L)$ coherent arithmetic operations to obtain the semi-digital encoding for Brownian motion, 
\al{ 
\sum_{\bm{a}} \sqrt{ p(\bm{a}) } \ket{\bm{a}} \otimes \frac{1}{\sqrt{T} } \sum_{t} \ket{t} \ket{B_{L}(\bm{a}, t)},
} 
where $B_{L}(\bm{a}, t)$ is the smoothed Brownian motion path obtained by truncating the Wiener series to $L$ terms given by coefficients $\bm{a} = (a_0, a_1, \dots, a_L)$. 

\STATE Coherently exponentiate the $B_{L}(\bm{a}, t)$ register, which gives us the following state
\all{ 
\sum_{\bm{a}} \sqrt{ p(\bm{a}) } \ket{\bm{a}} \otimes \frac{1}{\sqrt{T} } \sum_{t} \ket{t} \ket{G_{L, \mu, \sigma}(\bm{a}, t)},
} {eleven} 
where $G_{L, \mu, \sigma}(\bm{a}, t) = \exp(\sigma B_{L}(\bm{a}, t) + (\mu- \sigma^{2}/2) t)$.

\end{algorithmic}
\caption{Quantum algorithm for generating semi-digital encoding of geometric Brownian motion.}  \label{alg:sdGBM}

\end{algorithm}

\section{Quantum algorithm for Asian option pricing} \label{sec3}

In this section, we provide a quantum algorithm for pricing discretely monitored Asian options using the semi-digital encoding for the geometric Brownian motion. We begin with a discussion on Asian options on GBM and known methods to price them in Section~\ref{sec:asianintro}. This truncation analysis for pricing path dependent options using the KL expansion is carried out in Section~\ref{sec:trunk}, and the quantum Asian option pricing algorithm is presented and analyzed in Section~\ref{sec:Qasianoption}. 

\subsection{Asian options} \label{sec:asianintro}
Asian options are path dependent options whose payoff depends on, for example, the arithmetic or geometric average of the underlying asset over the period of the contract. In what follows, we will provide a quantum algorithm for pricing discretely monitored Asian call options in which the average taken is an arithmetic one, and the underlying follows a geometric Brownian motion. Discretely monitored Asian options are traded more frequently than continuously monitored Asian options due to their cost-effectiveness in trading and risk management.  
Asian options are less volatile than their vanilla counterparts by virtue of the fact that the payoff is based on an average of the underlying asset values over the time horizon of the contract, rather than simply the terminal asset value. As a result, Asian options are typically less expensive than their vanilla counterparts.

For continuously monitored Asian options with an arithmetic average, the expression for the price at time $0$ with strike $K$ is given as, 
\al{
\mathbb{E}\enc{\en{\frac{1}{T}\int_{0}^{T}S_u du-K}^{+}}
} 
a discounting factor $e^{-rT}$ may be added where $r$ is the riskless rate of interest. Through out this paper, we assume $r=0$ for simplicity while also noting that all methods discussed herein generalize straightforwardly to the case where a non-trivial discounting factor is required.

Methods for pricing continuously monitored Asian options with arithmetic averages include numerical estimation using probability density functions, Laplace transforms, approximation of the lognormal distribution, path integral approaches, and numerical methods for partial differential equations~\cite{geman1993,linetsky2004}. Note that there is an analytical solution for the price of a continuously monitored Asian option with a geometric average of the underlying price. 

The prices of discretely monitored Asian options on arithmetic averages do not have a closed form solution. Methods for pricing such options include Monte Carlo methods, numerical methods for partial differential equations, binomial lattice methods, analytic approximation to the lognormal distribution and approximation via Taylor Series expansion.

The price for an Asian option with strike $K$ on a discretely monitored stochastic process $S(t)$ is given by the expression, 
\al{
\mathbb{E} \left[ \left( \frac{1}{T} \sum_{t} S(t) - K\right)^+\right].
}
The expectation is over the different trajectories of the stochastic process. 
As mentioned in the previous section, we will consider the case where the underlying stochastic process is a GBM.

The standard classical approach for pricing such an Asian option using Monte Carlo methods has a complexity that scales as $O(T/\epsilon^2)$. Here, we propose quantum and quantum-inspired algorithms for pricing such options with only polylogarithmic dependence on $T$ at the expense of a worse dependence on $\epsilon$.

\subsection{Pricing path-dependent options with mapped Karhunen-Lo\`{e}ve expansion}
\label{sec:trunk}

In this section, we attempt to use exponentiated Wiener series to approximate the geometric Brownian motion and apply it to the pricing of Asian options.
We utilize results from Section~\ref{sec:mapped-kl-expansion} to analyze the error in the estimation of the option price due to the truncation in the Wiener series.

We first perform the analysis in a more general setting and then apply it to the specific case of pricing Asian options on an underlying modeled by a GBM.
The type of options we consider are path-dependent options that generalizes the Asian option. 
Specifically, we assume that the payoff function $f$ depends on the value of the underlying stochastic process $S_t$ at times $\{t_1,\dots,t_T\}$ that are discrete points in $[0,1]$. We consider continuous payoff functions, for which we make the following natural assumption, as stated in Assumption~\ref{asm:lipschitz-payoff}. 
\begin{assumption}
    \label{asm:lipschitz-payoff}
    For any two realizations $(x_{t_i})_{i=1}^{T}, (y_{t_i})_{i=1}^{T}$ of the underlying process $(S_{t_i})_{i=1}^{T}$, the payoff $f$ satisfies $\lvert f(x_{t_1}, \dots, x_{t_T}) - f(y_{t_1}, \dots, y_{t_T})\rvert \leq \sum_{i=1}^{T} w_{t_i} \lvert x_{t_i} - y_{t_i}\rvert$ for $\sum_i w_{t_i} =1$ with $w_{t_i} \geq 0$ almost surely. 
\end{assumption}
We state in Theorem~\ref{thm:asian-kl-error} our result on the pricing of such path-dependent options with the underlying modeled by a stochastic process that generalizes the GBM.
However, before that, we state a slightly more general lemma for the convenience of a later usage.
\begin{lemma}
\label{lem:asian-kl-error}
Let $S_t$ and $S'_t$ be two stochastic processes on $t \in [0, 1]$ with an $\mathcal{L}^2(\mathbb{P})$ distance $\mathbb{E} \enc{{\en{S_{t_i} - S'_{t_i}}}^2} \le \epsilon^2$.
Let $f$ be a payoff function satisfying Assumption~\ref{asm:lipschitz-payoff}.
Then the mean squared error between the payoffs of $\{S'_{t_i}\}_{i=1}^{T}$ and  $\{S_{t_i}\}_{i=1}^{T}$ is also bounded by $\epsilon^2$. 
In other words,
\al{
\mathbb{E} \enc{\en{f\en{S_{t_1}, \dots, S_{t_T}} - f\en{S'_{t_1}, \dots, S'_{t_T}}}^2} \le \epsilon^2.
}
\end{lemma}
\begin{proof}
From Assumption~\ref{asm:lipschitz-payoff} we have
\al{
\mathbb{E} \enc{\en{f\en{S_{t_1}, \dots, S_{t_T}} - f\en{S'_{t_1}, \dots, S'_{t_T}}}^2} 
\le \mathbb{E} \enc{\en{\sum_{i=1}^T w_i \abs{S_{t_i} - S'_{t_i}}}^2}.
}
Using our assumptions on the $w_i$ and the convexity of the square function, we have
\al{
\mathbb{E} \enc{\en{\sum_{i=1}^T w_i \abs{S_{t_i} - S'_{t_i}}}^2}
\le \mathbb{E} \enc{\sum_{i=1}^T w_i {\en{S_{t_i} - S'_{t_i}}}^2}
\le \max_{i} \mathbb{E} \enc{{\en{S_{t_i} - S'_{t_i}}}^2} \le \epsilon^2,
}
completing the proof.
\end{proof}

\begin{theorem}
\label{thm:asian-kl-error}
Let $S_t$ be a stochastic process on $t \in [0, 1]$ satisfying Assumption~\ref{asm:mapped-kl-process}, and $S'_t$ be the corresponding stochastic process with the KL series of $X_t$ truncated to $L_X(\epsilon)$ terms, where $X_t$ is the stochastic process used to construct $S_t$.
Let $f$ be a payoff function satisfying Assumption~\ref{asm:lipschitz-payoff}.
Then the mean squared error in approximated path-wise payoff $f$ using $S'_t$ compared to its true value using $S_t$ is $O(\epsilon^2)$. In other words,
\al{
\mathbb{E} \enc{\en{f\en{S_{t_1}, \dots, S_{t_T}} - f\en{S'_{t_1}, \dots, S'_{t_T}}}^2} = O(\epsilon^2).
}
\end{theorem}
\begin{proof}
Based on the assumptions on $S_t$ and $S'_t$, the  result from Theorem~\ref{thm:mapped-kl-truncation-error} applies, i.e.,
\al{
\mathbb{E} \enc{{\en{S_{t_i} - S'_{t_i}}}^2} = O(\epsilon^2).
}
We can then apply Lemma~\ref{lem:asian-kl-error}, from which the desired result is obtained immediately.
\end{proof}

For the case of Asian option pricing, the weight vector $(w_{t_i})_{i=1}^{T}$ is the uniform vector $w_{t_i}=\frac{1}{T}\, \forall i \in [T]$, which satisfies Assumption~\ref{asm:lipschitz-payoff}.
Furthermore, Theorem~\ref{thm:asian-kl-error} shows that when the underlying process is a GBM, retaining $O(1/\epsilon^{2})$ terms in the Wiener series suffices to obtain an $\epsilon$-approximation for the Asian option price. 

\subsection{Quantum algorithm using semi-digital encoding}\label{sec:Qasianoption}

The quantum algorithm for pricing an Asian option is specified as Algorithm~\ref{alg:AsianOption}. It uses a semi-digital encoding for the geometric Brownian motion 
with truncation parameter $L$ determined by the analysis in Theorem~\ref{thm:asian-kl-error} and then uses two nested standard amplitude estimation to estimate the 
Asian option payoff to additive error $\epsilon$. 

\begin{algorithm} [H]
\label{alg:quantum-kl-asian}
\begin{algorithmic}[1]
\REQUIRE Parameters $(S_0, \mu, \sigma)$ for geometric Brownian motion $S_t$ on $t \in [0, 1]$. Strike price $K$ for the Asian option and number of time steps $T$ for discrete monitoring. Accuracy $\epsilon$. 
\ENSURE  An additive error $\epsilon$ estimate of the Asian option price.  
\STATE Use the procedure detailed in step 1 of Algorithm~\ref{alg:sdGBM} prepare $L= O(1/\epsilon^{2})$ independent copies of Gaussian states, 
       \al{ 
        \sum_{\bm{a}} \sqrt{ p(\bm{a}) }  \ket{\bm{a}}. 
        } 
\STATE Let $U_{1}$ be the unitary that first applies steps 2-4 of Algorithm~\ref{alg:sdGBM} to create a semidigital encoding of geometric Brownian motion with 
drift and variance parameters $(\mu, \sigma)$.  It then applies a controlled rotation on an auxiliary qubit to get the geometric Brownian motion values 
in the amplitudes to prepare state, 
 \al{ 
        \label{eq:semi-digital-with-amp}
        \sum_{\bm{a}} \sqrt{ p(\bm{a}) } \ket{\bm{a}} \otimes \frac{1}{\sqrt{T} } \sum_{t} \ket{t} \ket{G_{L}(\bm{a}, t)} \en{  \sqrt{\frac{G_{L}(\bm{a}, t)}{ G_{\max} }} \ket{0} + \alpha(\bm{a}, t) \ket{1} },
        } 
 where $G_{\max} \geq \max_{\bm{a},t} G_{L}(\bm{a}, t)$ is a chosen normalization constant to ensure amplitudes on the last qubit have norms no greater than $1$, and $\alpha(\bm{a}, t) = \sqrt{1 - {G_{L}(\bm{a}, t)} / { G_{\max} }}$.
 
 \STATE Perform coherent amplitude estimation for the unitary $U_{1}$ to estimate to additive error $\epsilon$ the squared amplitude for the auxiliary qubit being $\ket{0}$ to obtain the quantum state, 
 \al{ 
 \label{eq:asian-inner-ae}
  \sum_{\bm{a}} \sqrt{ p(\bm{a}) } \ket{\bm{a}} \ket{\frac{\sum_{t} G_{L}(\bm{a}, t)}{T G_{\max}}}
 }    
 The extra registers are uncomputed by applying the inverse unitary $U_{1}^{-1}$. 

\STATE Let $U_{2}$ be the unitary that computes the Asian option payoff function in the auxiliary register and then applies a controlled rotation on an auxiliary qubit to get the 
Asian option payoff function in the amplitudes to prepare the state, 
\all{ 
  \sum_{\bm{a}} \sqrt{ p(\bm{a}) } \ket{\bm{a}} \ket{ \en{ \frac{\sum_{t} G_{L}(\bm{a}, t)}{T G_{\max}} - \frac{K}{G_{\max}} }^{+} } \en{  \sqrt{ \en{ \frac{\sum_{t} G_{L}(\bm{a}, t)}{T G_{\max}} - \frac{K}{G_{\max}} }^{+}} \ket{0} + \beta \ket{1} } 
 }  {eq:asian}  
 
 \STATE Perform amplitude estimation for the nested unitary $U= U_{1} U_{2}$ to get an additive error $G_{\max}^{-1}\epsilon$ estimate of the probability of measuring $\ket{0}$ in \eqref{eq:asian} which, when scaled by $G_{\max}$ is equal to the payoff for the Asian option to additive error $\epsilon$.

\end{algorithmic}
\caption{Quantum algorithm for Asian option pricing using semi-digital encodings.}  \label{alg:AsianOption}

\end{algorithm}

The error for the algorithm is determined by the truncation analysis in Theorem~\ref{thm:asian-kl-error}. The quantity being estimated in step 3 of the algorithm is the exponentiated average $\frac{1}{T} \sum_{t} f(B_{t})$ 
for the exponential function and Theorem~\ref{thm:asian-kl-error} shows that to get an $\epsilon$ additive error approximation for the average it suffices to retain $O(L^{2})$ terms in the Wiener series. 

When we use amplitude estimation to get to the state in \eqref{eq:asian-inner-ae}, the success probability for the procedure is $\gamma^2 > .81$. The success probability 
can be boosted to be arbitrarily close to one by taking a logarithmic number of independent copies and taking the median estimate, which ensures we can suppress the failure probability with a logarithmic overhead. Also if the auxiliary registers are not uncomputed in \eqref{eq:asian-inner-ae}, we can apply amplitude estimation again to estimate the combined amplitude of the $\ket{0}\ket{0}\ket{0}$ state for all $\bm{a}$.

The end-to-end algorithm for pricing the Asian option involves possible sources of error from the Wiener series expansion and two nested amplitude estimations, and the latter of which also introduce a failure probability, which can be boosted as indicated above. We have the following result on the overall complexity of the quantum Asian option pricing algorithm, 
\begin{theorem} 
Algorithm~\ref{alg:AsianOption} estimates the payoff for a discretely monitored Asian option up to additive error $\epsilon<1/T$ using $O(\epsilon^{-2})$ qubits and with overall complexity $O(\epsilon^{-4}\log^2(\delta^{-1})\polylog(T))$, where $\epsilon$ is the desired relative error and $\delta$ is the failure probability of the algorithm.
\end{theorem} 

The high complexity of the semi-digital encoding based approach for Asian option pricing is due to the use of two nested amplitude estimation procedures and this arises from the specific form of the Asian option payoff function. 
One way of obtaining asymptotic improvements for the algorithm is to improve the KL expansion truncation error analysis in Theorem~\ref{thm:asian-kl-error} (or more precisely Lemma~\ref{lem:asian-kl-error}) by better utilizing properties of the payoff function.
However, we obtain a quantum Asian option pricing algorithm with complexity $O(1/\epsilon^3)$ by an alternative quantum-inspired method in section \ref{sec:quantized-algorithms}.

\section{Classical algorithms inspired by the semi-digital encoding approach}\label{sec:classical}

In this section, we consider classical algorithms inspired by the quantum semi-digital encoding approach presented in Section~\ref{sec3} for the option pricing problems considered in this paper.  
For the Asian option pricing problem, we provide two different algorithms whose running time has a poly-logarithmic dependence on $T$. 
The baseline classical algorithm to compare against has running time $O(T\sigma/\epsilon^{2})$, where $\sigma$ is the standard deviation of the option payoff.
The baseline algorithm generates the stochastic process trajectory in time $O(T)$ and then uses the classical Monte Carlo method for estimating the option price up to a root mean squared error (RMSE) of $\epsilon$, or equivalently, an additive error of $\epsilon$ with high probability. 
\subsection{Quantum-inspired sampling}

Algorithm~\ref{alg:quantum-kl-asian} suggests us that in order to approximately compute the pathwise mean of the underlying asset price modeled by a GBM $G(t)$ across all monitored time points $t \in [0, 1]$ for the Asian option payoff, it suffices to prepare classical sample query access for $G_L(\bm{a}, t)$ as we did in Equation~\eqref{eq:semi-digital-with-amp} for the quantum algorithm, and draw random samples from it to estimate the mean.
Specifically, we wish to sample $G_L(\bm{a}, t)$ at different $t$ with probabilities proportional to $G_L(\bm{a}, t)$.

This is easily accomplished by rejection sampling as follows: 
\begin{enumerate}
    \item Sample $t$ uniformly from $[0, 1]$.
    \item Sample $z$ uniformly from $[0, 1]$.
    \item Accept the sample $t$ iff $z \le G_{L}(\bm{a},t)/G_{\max}$.
\end{enumerate}
It is easy to see that the distribution over $t \in [0, 1]$ conditioned on the acceptance criteria is the correct one, and the expected number of trials required to generate one successful sample $t$ across all $t \in [0, 1]$ is given by $\frac{ T G_{\max} }{\sum_t G_{L}(\bm{a}, t)}$. 
By standard reduction, we can use $\frac{ T G_{\max} }{\sum_t G_{L}(\bm{a}, t)} \log(1/\eta)$ samples to obtain a sample from a distribution that has total variation distance at most $\eta$ from the correct distribution. 
These samples can then be used as in \cite[Algorithm 1]{tang2021pca} to price any option whose payoff is a function of the time average of the underlying asset price.

Algorithm~\ref{alg:asian-randomized-sampling} prices a discretely monitored Asian option on GBM using this approach. The same algorithm can be applied to approximate the continuously monitored case as well, however for continuously monitored Asian options on the GBM, it may be possible to obtain analytically closed form solutions. 
\begin{algorithm} [H]
\begin{algorithmic}[1]
\REQUIRE Parameters $(S_0, \mu, \sigma)$ for geometric Brownian motion $S_t$ on $t \in [0, 1]$. Strike price $K$ for the Asian option and number of time steps $T$ for discrete monitoring. Accuracy $\epsilon$. 
\ENSURE  An additive error $\epsilon$ estimate of the Asian option payoff function.  
\STATE Generate $L+1$ random samples $\bm{a} = (a_0, a_1, \dots, a_L)$ from the standard normal distribution.
\STATE Use the rejection sampling approach described above to sample $M_1= O(1/\epsilon^{2})$ points from $t \in [0, 1]$ with probability proportional to $G_L(\bm{a}, t)$.
\STATE Compute the average of $G_L(\bm{a}, t)$ at the sampled $t$ values and compute the pathwise payoff using this average, i.e,
\al{
    \en{\frac{1}{M_1} \sum_{j=1}^{M_1} G_L(\bm{a}, t_j) - K}^+.
}
\STATE Repeat steps 1--3 $M_0= O(1/\epsilon^{2})$ times, and compute the average of the pathwise payoff
\al{
    \frac{1}{M_0} \sum_{\bm{a}} \en{\frac{1}{M_1} \sum_{j=1}^{M_1} G_L(\bm{a}, t_j) - K}^+. \label{eq:asian-nested-mc-estimator}
}
This is our estimator for the Asian option price.
\end{algorithmic}
\caption{Quantum-inspired sampling for Asian option pricing} \label{alg:asian-randomized-sampling} 

\end{algorithm}

It is easy to see that the estimator given by Equation~\eqref{eq:asian-nested-mc-estimator} converges to the following value as $M_0 \to \infty$ and $M_1 \to \infty$
\al{
    \mathbb{E}_{\bm{a}}\enc{\en{\bar{G_L}(\bm{a}) - K}^+}, \label{eq:asian-nested-mc-true-mean}
}
where $\bar{G_L}(\bm{a}) = \int_{0}^{1} G_L(\bm{a}, t) dt$ is the mean of $G_L(\bm{a}, t)$ on $t \in [0, 1]$.

Since the estimator for the option price given by Equation~\eqref{eq:asian-nested-mc-estimator} involves a two-level nested Monte Carlo sampling, the mean squared error $\epsilon^2$ of the estimator w.r.t. its true mean~\eqref{eq:asian-nested-mc-true-mean} generally scales asymptotically as $O(\sigma^2/M_0 + \sigma^2_1/M_1)$, where $\sigma$ and $\sigma_1$ are the standard deviations of variables that the two estimators are estimating the means of, i.e., the payoff $\en{\bar{G_L}(\bm{a}) - K}^+$ and the time average $\bar{G_L}(\bm{a})$.
Therefore, to estimate Equation~\eqref{eq:asian-nested-mc-true-mean} up to $\epsilon$ RMSE, the total number of samples required in the above nested Monte Carlo simulation is $M_0 M_1 = O(\sigma/\epsilon^4)$, where we have dropped the factor $\sigma_1$ for conciseness and consistency with the analysis for other algorithms presented in this paper.
In the case where $T$ is large, 
\al{
\mathbb{E}_{\bm{a}} \left[ \left( \frac{1}{T} \sum_{t} G_L(\bm{a}, t) - K\right)^+\right] = \mathbb{E}_{\bm{a}}\enc{\en{\bar{G_L}(\bm{a}) - K}^+} + O\en{\frac{1}{T}}. \label{eq:asian-discrete-truncated-mean}
}
Therefore, when $T \gg 1/\epsilon$, the total number of samples required to estimate Equation~\eqref{eq:asian-discrete-truncated-mean} also scales as $O(\sigma/\epsilon^4)$.

Additionally, for each sample $t$, a degree $L$ polynomial needs to be evaluated. This would add an additional multiplicative factor of $L$ to the $M_0 M_1$ complexity.
Also note that a total of $(L+1)M_0$ Gaussian samples needs to be generated for the $M_0$ paths. However, this term is dominated by the $L M_0 M_1$ term above.
According to Theorem~\ref{thm:asian-kl-error}, we should choose $L$ to be $O(1/\epsilon^{2})$ in order to guarantee an approximation error of at most $\epsilon$ in the Asian payoff process from the truncated Wiener series. Therefore, the overall complexity of this algorithm is $O(\sigma/\epsilon^{6})$.

We note that the above convergence analysis for the nested Monte Carlo is very general, therefore it may very likely be too pessimistic for the specific problem we are considering and the actual convergence in practice may be much faster. 
Specifically, if we assume that inner estimator $\bar{G_L}(\bm{a})$ is approximately normally distributed, it can be shown that the mean squared error $\epsilon^2$ of the estimator~\eqref{eq:asian-nested-mc-estimator} w.r.t. its true mean~\eqref{eq:asian-nested-mc-true-mean} scale asymptotically as $O(\sigma^2/M_0 + \sigma_1^2/M_1^2)$ as opposed to $O(\sigma^2/M_0 + \sigma^2_1/M_1)$~\cite{hong2009estimating}. 
This consequently reduces the total number of samples required to estimate Equation~\eqref{eq:asian-nested-mc-true-mean} up to $\epsilon$ RMSE to $M_0 M_1 = O(\sigma^2/\epsilon^3)$, and the overall complexity of the pricing algorithm to $O(\sigma/\epsilon^{5})$.

In addition, one may also utilize variance reduction techniques such as multi-level Monte Carlo (MLMC) to reduce the number of samples required to achieve the same RMSE.
This will be discussed in more detail in Section~\ref{sec:mlmc}.

\subsection{Time-domain sub-sampling}
\label{sec:subsampling}

In this section, we attempt to analyze a time-domain sub-sampling-based algorithm for pricing certain path-dependent options.
The main result is that cost of generating one sample for the payoff may be assumed to scale as $O(1/\epsilon^{2})$ for the cases of Brownian motion and geometric Brownian motion, and as $\poly{1/\epsilon}$ under fairly general regularity assumptions on the stochastic process and the payoff. 

As in Section~\ref{sec:trunk}, we assume that the payoff is a function satisfying Assumption~\ref{asm:lipschitz-payoff}, and the underlying asset price $S_t$ is modeled by a function $g$ of another stochastic process $X_t$ with a Karhunen-Lo\`{e}ve expansion as specified in Assumption~\ref{asm:mapped-kl-process}.
Additionally, we make the following further assumption on the regularity of the stochastic process $X_t$.
\begin{assumption}
\label{asm:bounded-eigenfunctions}
    Let $k_X(s,t)$ be the Mercer kernel associated to stochastic process $X_t$ and $\mathcal{K}$ be the corresponding integral operator with eigenfunctions $\{e_k\}$ and eigenvalues $\{\lambda_k\}$. We assume the following:
    \begin{enumerate}
        \item The eigenfunctions $\{e_k\}$ are Lipschitz continuous on $[0,1]$,  i.e. $\forall k \in \mathbb{N} \; \exists G(k) \in \mathbb{R}_{+}$ such that $ \forall t,s \in [0,1], |e_k(t) - e_k(s)| \le G(k)\lvert t-s \rvert$.
        \item  $\lambda_{k}G(k)^{2} \leq C_\mathrm{M}, \forall k$ and for a universal constant $C_\mathrm{M}$.
    \end{enumerate}
\end{assumption}
Note that the second condition is met by Brownian motion, where the squared Lipschitz constant for the eigenfunctions is matched by the decay rate for the eigenvalues. Theorem~\ref{thm:specBM} shows that the Lipschitz constant for the Brownian motion Mercer kernel eigenfunctions is $O(k)$ and the eigenfunctions $\lambda_{k} = O(1/k^{2})$. Therefore, Assumption~\ref{asm:bounded-eigenfunctions} is satisfied for Brownian motion. 

The existence of the truncation index $L_{X}(\epsilon)$ ensures that the stochastic process $X_t$ is approximately Lipschitz continuous in expectation, in the following sense.
\begin{lemma}
    \label{lem:KL-smooth}
    Let $X_t$ be a stochastic process satisfying Assumption~\ref{asm:bounded-eigenfunctions} and $\epsilon > 0$ be any given error parameter. Let $L_X(\epsilon)$ be the corresponding truncation index as defined in Definition~\ref{def:trunc-index}. It then holds that,
    \begin{align}
        \mathbb{E} \enc{\en{X_t - X_s}^2} \le 3 C_\mathrm{M} L_X(\epsilon) (t - s)^2 + 6 \epsilon^2, \quad \forall t,s \in [0,1], 
    \end{align}
    where $C_\mathrm{M} = \sup_{k \in \mathbb{N}} \lambda_kG^2(k)$ is a universal constant as specified in Assumption~\ref{asm:bounded-eigenfunctions}.
\end{lemma}
\begin{proof}
By the definition of truncation index, we have that $\mathbb{E}\left[\right(X'_t - X_t\left)^2\right] \le \epsilon^2$, where $X'_t = \sum_{k=0}^{L_X(\epsilon)} Z_k e_k(t)$. It then follows that
\al{
\mathbb{E} \enc{\en{X_t - X_s}^2} 
&= \mathbb{E} \enc{\en{X'_t - X'_s + X_t - X'_t + X_s - X'_s}^2} \\
&\le 3 \en{\mathbb{E} \enc{\en{X'_t - X'_s}^2} + \mathbb{E} \enc{\en{X_t - X'_t}^2} + \mathbb{E} \enc{\en{X_s - X'_s}^2} } \\
&\le 3 \en{\mathbb{E} \enc{\en{X'_t - X'_s}^2} + 2\epsilon^2},
}
where first inequality follows from Cauchy-Schwartz.
To bound the first term on the last line, we notice that
\begin{align}
    \mathbb{E} \enc{\en{X'_t - X'_s}^2}
    &= \mathbb{E}\left[\left(\sum_{k=0}^{L_X(\epsilon)} Z_k (e_k(t)-e_k(s))\right)^2\right] \\
    &\le \sum_{k,l=0}^{L_X(\epsilon)}\mathbb{E}(Z_k Z_l) (e_k(t)-e_k(s))(e_l(t) - e_l(s)) \\
    &\le \sum_{k=0}^{L_X(\epsilon)} \lambda_k (e_k(t) - e_k(s))^2 \\
    &\le \sum_{k=0}^{L_X(\epsilon)} \lambda_k G^{2}(k) (t - s)^2 \\
    &\le L_X(\epsilon) C_\mathrm{M} (t - s)^2,
\end{align}
where $C_\mathrm{M} = \sup_{k \in \mathbb{N}} \lambda_kG^2(k)$ is a universal constant per Assumption~\ref{asm:bounded-eigenfunctions}.
It then follows that
\al{
\mathbb{E} \enc{\en{X_t - X_s}^2} \le 3 C_{M} L_X(\epsilon) (t - s)^2 + 6 \epsilon^2,
}
hence completing the proof. 
\end{proof}
\noindent In particular, for the case of Brownian motion $B_t$, the validity of Assumption~\ref{asm:bounded-eigenfunctions} follows from Theorem~\ref{thm:specBM}. 
Applying Lemma~\ref{lem:KL-smooth} with $L_B(\epsilon)$ from Theorem~\ref{thm:LXBrownian} and noting that $C_\mathrm{M} = 1$ we have $\mathbb{E} \enc{\en{B_t - B_s}^2} \le 3(t-s)^2/\epsilon^2 + 6 \epsilon^2$. 

With Assumption~\ref{asm:mapped-kl-process} and particularly Assumption~\ref{asm:sub-exp-lipschitz-func} enclosed therein, we now generalize Lemma~\ref{lem:KL-smooth} to the stochastic process $S_t$ for the underlying asset price.
To do that, we need another assumption on the smoothness of $g_t$ in $t$.
\begin{assumption}
\label{asm:lipschitz-func-time}
The function $g_t(x)$ is Lipschitz continuous in $t$ with Lipschitz constant $K'_g$.
\end{assumption}
\noindent And the generalized lemma is stated as follows.
\begin{lemma}
    \label{lem:func-of-KL-smooth}
    Let $S_t=g_t(X_t)$ be a stochastic process as specified in Assumption~\ref{asm:mapped-kl-process} and additionally assume that $X_t$ and $g_t$ satisfy Assumptions~\ref{asm:bounded-eigenfunctions} and \ref{asm:lipschitz-func-time} respectively. Let $\epsilon > 0$ be any given error parameter and $L_X(\epsilon)$ be the truncation index as defined in Definition~\ref{def:trunc-index} for $X_t$. It then holds that,
    \begin{align}
        \mathbb{E} \enc{\en{S_t - S_s}^2} \le C_3\en{ \en{L_X(\epsilon) + {K'_g}^2} (t - s)^2 + \epsilon^2}, \quad \forall t,s \in [0,1], 
    \end{align}
    where $C_3$ is a universal constant. 
\end{lemma}
\begin{proof}
By the construction of $S_t$, we have
\al{
\mathbb{E} \enc{\en{S_t - S_s}^2} 
&= \mathbb{E} \enc{\en{g_t(X_t) - g_s(X_s)}^2} \\
&= \mathbb{E} \enc{\en{g_t(X_t) - g_t(X_s) + g_t(X_s) - g_s(X_s)}^2} \\
&\le 2\mathbb{E} \enc{\en{g_t(X_t) - g_t(X_s)}^2} + 2\mathbb{E} \enc{\en{g_t(X_s) - g_s(X_s)}^2}.
}
We can bound the first term on the r.h.s. of the inequality above using Theorem~\ref{thm:mapped-kl-truncation-error} and Lemma~\ref{lem:KL-smooth}
\al{
\mathbb{E} \enc{\en{g_t(X_t) - g_t(X_s)}^2} \le C_2 \en{3 C_{M} L_X(\epsilon) (t - s)^2 + 6 \epsilon^2}.
}
Additionally, since $g_t$ is $K'_g$-Lipschitz in $t$,
\al{
\mathbb{E} \enc{\en{g_t(X_s) - g_s(X_s)}^2} \le {K'_g}^2 (t - s)^2.
}
Therefore
\al{
\mathbb{E} \enc{\en{S_t - S_s}^2} &\le 2 C_2 \en{3 C_{M} L_X(\epsilon) (t - s)^2 + 6 \epsilon^2} + 2{K'_g}^2 (t - s)^2 \\
&= C_3\en{ \en{L_X(\epsilon) + {K'_g}^2} (t - s)^2 + \epsilon^2},
}
where $C_3 = \max(6C_2 \max (C_M, 2),2)$, completing the proof.

\end{proof}
In particular, the geometric Brownian motion modeled by exponentiating the KL expansion of the Brownian motion (i.e., the Wiener series) satisfies Assumptions~\ref{asm:mapped-kl-process}, \ref{asm:bounded-eigenfunctions} and \ref{asm:lipschitz-func-time}, and hence has the smoothness property given by Lemma~\ref{lem:func-of-KL-smooth}.

With the above assumptions and lemmas, we can now describe a time-domain sub-sampling algorithm with $M < T$ points as follows: suppose the payoff function depends on times $\{t_1,\dots,t_T\}$ and define $c(t)$ for any $t \in [0,1]$ by rounding down to an element of set $\{i/M \mid 0 \le i < M\}$, that is $c(t) = \lfloor tM \rfloor/M$. 
The sub-sampling approach uses the estimator $f(S_{c(t_1)},\dots,S_{c(t_T)})$ for $\mathbb{E}\left[f(S_{t_1},\dots,S_{t_T})\right]$. 
Our main result is the following:
\begin{theorem}\label{thm:sub_s}
Let $f$ be a payoff function that is dependent on $T$ discrete points and satisfies Assumption~\ref{asm:lipschitz-payoff}.
Let $S_t = g_t(X_t)$ be the stochastic process on which the payoff is evaluated, satisfying Assumption~\ref{asm:mapped-kl-process}, with $X_t$ and $g_t$ additionally satisfying Assumptions~\ref{asm:bounded-eigenfunctions} and \ref{asm:lipschitz-func-time} respectively. 
A sub-sampling estimator with $M=\max \en{\frac{\sqrt{L_X(\epsilon)}}{\epsilon}, \frac{{K'_g}^2}{\epsilon^2}}$ points as described above, is an estimator for $\mathbb{E}\left[f(S_{t_1},\dots,S_{t_T})\right]$ with a mean squared error of $O(\epsilon^2)$.
\end{theorem}
\begin{proof}
From Lemmas~\ref{lem:asian-kl-error} and \ref{lem:func-of-KL-smooth}, we have
\al{
\mathbb{E}\left[\en{f(S_{c(t_1)},\dots,S_{c(t_T)}) -   f(S_{t_1},\dots,S_{t_T})}^2\right] \le C_3\left(\en{L_X(\epsilon)+{K'_g}^2} ( t_i - c(t_i) )^2 + \epsilon^2\right).
}
This implies that if the sub-sampling estimator uses $\max \en{\frac{\sqrt{L_X(\epsilon)}}{\epsilon}, \frac{{K'_g}^2}{\epsilon^2}}$ equally spaced points in time, we have for all $1 \le i \le T, \lvert t_i - c(t_i) \rvert \le \min \en{\frac{\epsilon}{\sqrt{L_X(\epsilon)}}, \frac{\epsilon^2}{{K'_g}^2}}$, and consequently
\al{
C_3\left(\en{L_X(\epsilon)+{K'_g}^2} ( t_i - c(t_i) )^2 + \epsilon^2\right) \le 3C_3\epsilon^2,
}
completing the proof.
\end{proof}

There are only $M$ distinct values to be summed up to evaluate the estimator for the payoff $f$ on each path. 
We make a further natural assumption that the stochastic process is \emph{fast-forwardable}.
In other words, we assume that the process is Markovian and that a closed-form distribution is known for any future time point given the value at the current time. 
For such processes, each estimator for $f$ can be evaluated in time $O(M)$.
By using classical Monte Carlo methods, an estimator for the mean of the payoff (i.e., the option price) can be obtained by taking $O(1/\epsilon^2)$ samples of the estimator for $f$, to achieve a mean squared error of $O(\epsilon^2)$.
We present this result formally in the following theorem.
\begin{theorem} \label{thm:MCsample}
There is a classical Monte Carlo algorithm with complexity $\max \en{\frac{\sqrt{L_X(\epsilon)}}{\epsilon^3}, \frac{{K'_g}^2}{\epsilon^4}}$ for estimating $E[f(S_{t})]$ to mean squared error $\epsilon^2$ for fast-forwardable stochastic process $S_t$ satisfying Assumption~\ref{asm:mapped-kl-process} and continuous payoff function $f$ satisfying Assumption~\ref{asm:lipschitz-payoff}.
\end{theorem}

On the other hand, a pricing algorithm based on KL expansion similar to Algorithm~\ref{alg:asian-randomized-sampling} requires $O(L_X(\epsilon)/\epsilon^4)$ complexity.
Theorem~\ref{thm:MCsample} then implies that the time-domain sub-sampling approach is more efficient in such cases.

In particular, Theorem~\ref{thm:MCsample} applies to the Asian option pricing problem on $t \in [0,1]$, where the underlying stochastic process $S_t$ is a geometric Brownian motion approximated by exponentiating the KL expansion of the Brownian motion $X_t$.
In that case, the truncation index at error $\epsilon$ for the Brownian motion $X_t$ is $L_X(\epsilon) = O(1/\epsilon^2)$ as given by Theorem~\ref{thm:LXBrownian}. 
Therefore, using $O(1/\epsilon^2)$ grid points in the time domain is sufficient to approximate the time average of $S(t)$ up to an error of at most $\epsilon$.
This yields a classical Monte Carlo algorithm with complexity $\widetilde{O}(1/\epsilon^4)$ for pricing an Asian option on an underlying asset modeled by GBM. The algorithm is presented in Algorithm~\ref{alg:asian-subsampling}.

\begin{algorithm} [H]
\begin{algorithmic}[1]
\REQUIRE Parameters $(S_0, \mu, \sigma)$ for geometric Brownian motion $S_t$ on $t \in [0, 1]$. Strike price $K$ for the Asian option and number of time steps $T$ for discrete monitoring. Accuracy $\epsilon$. 
\ENSURE  A root mean squared error $\epsilon$ estimate of the Asian option payoff function. 
\STATE Construct an $(L+1)$-point grid in the time domain $[0, 1]$ denoted by time points $t_0=0, t_1=1/L, \dots, t_L=1$, where $L = 1/\epsilon^{2}$.
\STATE For each $k \in [L]$, generate a sample for $S_k$ by sampling from $\mathrm{Lognormal}\en{\en{\mu-\frac{\sigma^2}{2}}t_k, \sigma^2 t_k}$ and then multiply by $S_0$.
The generated $S_1, \dots, S_L$ form a sample path for the underlying asset price.
\STATE Evaluate the Asian option payoff function on the generated sample path $f(S_1, \dots, S_L) = \en{\frac{1}{L}\sum_{k\in[L]} S_k - K}^+$. 
Note that the payoff function has been (trivially) modified to take $L$ values instead of $T$ values in time.
\STATE Repeat steps 2 and 3 to generate $O(1/\epsilon^2)$ samples for the payoff $f(S_1, \dots, S_L)$, and compute the average of these samples.
This is our estimator for the price of the Asian option.

\end{algorithmic}
\caption{Time-domain sub-sampling for Asian option pricing} \label{alg:asian-subsampling}
\end{algorithm}

\section{Quantized time-domain sub-sampling}

\label{sec:quantized-algorithms}

In this section, we present a quantum Asian option pricing algorithm that is obtained by quantizing the classical Monte Carlo method described above in Theorem~\ref{thm:MCsample}. The quantum algorithm given as algorithm \ref{A4} has complexity $\widetilde{O}(1/\epsilon^{3})$ improving upon the quantum Asian option pricing algorithm using the semi-digital encodings.

The correctness of the algorithm \ref{A4} follows from the analysis for the classical Monte Carlo algorithm. The estimator in step 4 is biased with error $\epsilon$, an amplitude estimation with additive error $\epsilon$ in step 5 therefore recovers  $E_{\bm{a}} [ (\overline{G}_{\mu, \sigma}(a)- K)^{+}] $ to additive error $2\epsilon$. 

The number of qubits used by the algorithm is $\widetilde{O}(M)= O(1/\epsilon^{2})$ assuming all the calculations are carried out 
at constant precision. The number of arithmetic operations needed for step $2$ for exponentiation and maintaining a running average of the discretized Brownian motion is $O(M)$. The overall arithmetic complexity of steps 1-4 is therefore $\widetilde{O}(1/\epsilon^{2})$. Amplitude estimation in step 5 adds a multiplicative $O(1/\epsilon)$ overhead to the complexity. The resource requirements for the quantum algorithm are given by the following theorem. 

\begin{theorem} 
There is a quantum algorithm for estimating the Asian option payoff function $E_{\bm{a}} [ (\overline{G}_{\mu, \sigma}(a)- K)^{+}] $ to additive error $2\epsilon$ over geometric Brownian motion trajectories, the algorithm uses $\widetilde{O}(1/\epsilon^{2})$ qubits and has gate complexity $\widetilde{O}(1/\epsilon^{3})$. 
\end{theorem}

\begin{algorithm} [H]
\begin{algorithmic}[1]
\REQUIRE Number of qubits $n$ used for Gaussian state preparation, parameters $\mu, \sigma$ for GBM, accuracy $\epsilon$ and Strike price $K$ for the Asian option. 
\ENSURE  An additive error $\epsilon$ estimate of $\mathbb{E}[(\frac{1}{T} \sum S(t) - K)^{+}]$ where the expectation is over geometric Brownian motion paths $G_{\mu, \sigma}$. 
\STATE Let $M= O(\sqrt{L_{\epsilon}(X)}/\epsilon) = O(1/\epsilon^{2})$, prepare $M$ independent copies of the standard normal Gaussian state
on $N= 2^{n}$ qubits using the method of \cite{BGM22}. 

\al{ 
\ket{G}^{\otimes (M+1)} = \bigotimes_{k=0}^{M} \left( \sum_{a_k=-N/2}^{N/2} \sqrt{ p_k(a_k) }\ket{a_k} \right) 
=  \sum_{\bm{a}} \sqrt{ p(\bm{a}) } \ket{\bm{a}}.
}
\STATE Append two auxiliary registers and compute coherently on the first the running value for the Brownian motion $B(t) = \frac{1}{M} \sum a_{m}$ and on the second the running average for the GBM $G_{\mu, \sigma}(t) = S_{0} exp(\sigma B(t) + (\mu - \sigma^{2}/2) t)$. Uncompute the value $\ket{B(1)}$ in the first register to obtain, 

\al{ 
 \sum_{\bm{a}} \sqrt{ p(\bm{a}) } \ket{\bm{a}}  \ket{ \overline{G}_{\mu, \sigma}(a) } .
}
where $\overline{G}_{\mu, \sigma}(a)= \frac{1}{M} \sum_{i}  G_{\mu, \sigma} (i/M)$.

\STATE Threshold against the strike price $K$ and compute the positive part to obtain,  
\al{ 
  \sum_{\bm{a}} \sqrt{ p(\bm{a}) } \ket{\bm{a}}  \ket{ (\overline{G}_{\mu, \sigma}(a)  - K)^{+} } .
}

\STATE Append extra qubit and apply a conditional rotation to get the 
payoff function into the amplitude. Uncompute the $\ket{ (\overline{G}_{\mu, \sigma}(a)  - K)^{+} } $ to obtain, 
\al{ 
 \sum_{\bm{a}} \sqrt{ p(\bm{a}) } \ket{\bm{a}}  
 \left ( \sqrt{ \frac{(\overline{G}_{\mu, \sigma}(a)- K)^{+}}{G_{max}}}  \ket{0} + \en{ 1 - \frac{(\overline{G}_{\mu, \sigma}(a)- K)^{+}}{G_{max}}}^{1/2}   \ket{1} \right )   .
}
where $G_{max}$ is a normalizing factor to ensure that the amplitudes have norm at most $1$. 
\STATE Perform amplitude estimation to estimate the probability that the auxiliary qubit is in state $\ket{0}$ to recover an additive error $\epsilon$ estimate of the payoff $E_{\bm{a}} [ (\overline{G}_{\mu, \sigma}(a)- K)^{+}] $.

\end{algorithmic}
\caption{Quantum algorithm for Asian option pricing over geometric Brownian motion trajectories.} \label{A4}
\end{algorithm}

\section{Discussion}
\label{sec:discussion}

We have proposed a new type of encoding for stochastic processes, namely the semi-digital encoding and given quantum algorithms to generate semi-digital encodings for exponentiated Gaussian processes using the KL expansion. We have shown how our quantum algorithm can price discretely monitored Asian options. The approach can be extended to pricing of exotic, path dependent options such as Barrier Options. 
However, while it is easy to come up with quantum algorithms for these more exotic options using semi-digital encodings, more detailed analysis is required to compare with classical algorithms and to determine the extent of the quantum speedup.

A natural extension of our work would be to devise analogous algorithms for more sophisticated and realistic asset price models, such as L\'{e}vy processes, stable processes, the exponential Ornstein-Uhlenbeck process, exponentiated fractional Brownian motion, stochastic volatility models, just to name a few. One would naturally need to define a semi-digital encoding for these processes as a first step. 
One might also consider barrier options with a double barrier, a time-dependent barrier, lookback options, or Asian options with a geometric rather than arithmetic average. 

In the rest of this section, we discuss several aspects around the practicalities of the algorithms discussed in this paper, and the open problems that they pose.

\subsection{Practicalities of the Black-Scholes model}

We studied option pricing where the underlying asset is modeled by a smooth function of a sub-Gaussian process as stated in Assumption~\ref{asm:mapped-kl-process}.
We specifically focused on the geometric Brownian motion corresponding to solutions of the Black-Scholes equation. 
Despite its convenient closed-form solution, the Black-Scholes model suffers from many limitations. 
Some of these limitations arise from the assumption of constant volatility $\sigma$ in the underlying stochastic process, which fails to explain the empirically observed phenomena of volatility clustering, volatility smile, and volatility term structure.
The assumption of log-normal and independent returns is also often challenged by empirical data. For example, the Black-Scholes model results in consistent erroneous option prices when heavy tails are present in the underlying asset price distribution. Notwithstanding its aforementioned limitations, the Black-Scholes model is still one of the most widely used models in quantitative finance due to its simplicity and its well-established popularity in the construction and delta hedging of simple option products.
Therefore, the Black-Scholes model provides a production-relevant baseline for the type of quantum speedups that can be expected for option pricing problems. 

To address the limitations of the Black-Scholes model, a variety of more sophisticated models have been proposed with the addition of local~\cite{dupire1994pricing,derman1994riding,cox1975notes} or stochastic~\cite{heston1993,hagan2002managing} volatility terms.
We leave the extension of our techniques to these models to a future study.

\subsection{Practicalities of Karhunen-Lo\`{e}ve expansion in option pricing}

The Karhunen-Lo\`{e}ve-expansion-based approach to option pricing relies on knowledge of the eigenvectors and eigenvalues for the stochastic process. 
Although we have seen above that the Brownian motion has a simple KL expansion, it is not the case for a general stochastic process. The Kosambi-Karhunen-Lo\`{e}ve theorem (Theorem~\ref{thm:kkl}) applies generally to square-integrable stochastic processes, however, it may be difficult to obtain explicit solutions for the eigenvalues and eigenfunctions of every stochastic process that has a KL expansion. 
In other words, we cannot always obtain analytical solutions to the Fredholm integral equation 
\al {
\label{eq:cov} \int_{0}^{T} k_{X}(s,t) f(s)= \lambda f(t),
} 
where $k_{X}(s,t)=\mathrm{Cov}(X_s, X_t)$.
Therefore one often has to resort to numerical solutions. 

In addition to the Brownian motion, other examples where analytical solutions to Equation~\eqref{eq:cov} exist include the Brownian bridge, the fractional Brownian motion, the $m$-integrated Brownian motion, the Ornstein-Uhlenbeck process, the Ornstein-Uhlenbeck bridge, the Poisson homogeneous process and Gaussian processes with exponential and squared-exponential covariance kernels. In many of these cases, the eigenvectors are known functions like the sinusoidal functions or the Hermite polynomials, thus quantum pricing algorithms for these processes can be developed using analog or semi-digital encodings and the extent of quantum speedups can be analyzed precisely.

For a stochastic process that has a KL expansion with a known truncation index, Theorem~\ref{thm:sub_s}, links the performance of the time-domain sub-sampling algorithm with that of the KL-expansion-based algorithm.
We would like to note that Theorem~\ref{thm:sub_s} makes use of Assumption~\ref{asm:bounded-eigenfunctions}, which states that the eigenvalue times the squared Lipschitz constant of the corresponding eigenvector is bounded by a constant. 
This suggests an interesting open problem, that is to investigate stochastic processes with known KL expansion eigenvalues and eigenfunctions, but violate Assumption~\ref{asm:bounded-eigenfunctions}. 
This could potentially result in a worse running time than polylogarithmic for the time-domain sub-sampling approach, and hence opens up opportunities for larger speedups for quantum algorithms using the semi-digital encoding.

\subsection{Multi-level Monte Carlo}
\label{sec:mlmc}

The problem of pricing path-dependent options, such as Asian options as considered herein, using Monte Carlo methods involves a two-level sampling process. Specifically, the Monte Carlo method consists of a sampling in time nested in another sampling across realizations (i.e., sample paths) of the stochastic process.
As a result, such processes may benefit from variance reduction techniques such as multi-level Monte Carlo (MLMC)~\cite{giles2008multilevel}. 
Under certain conditions on the weak and strong orders of convergence of the time-domain sub-sampling scheme with an increasing number of time points, MLMC may reduce the sampling complexity of the two nested Monte Carlo estimators to that of a single one up to a polylogarithmic factor in the error tolerance.
A quantized version of the MLMC method has also been developed with similar speedups for quantum Monte Carlo mean estimation, under analogous assumptions on the convergence order of the time-domain sub-sampling scheme~\cite{an2021qmlmc}.

We note that with a few simple modifications, MLMC should also work for the KL-expansion-based pricing algorithms.
Specifically, at each numerical approximation level $l$, the approximation $P_l$ for the payoff may be constructed using a KL expansion with a truncation index of $L=2^l$. Then samples for the independent estimators $Y_l = P_l - P_{l-1}$ may be generated by sharing the randomness in the $2^{l-1}$ highest order terms of $P_l$ and $P_{l-1}$.
We leave the detailed analysis of using MLMC with the time-domain sampling and KL-expansion approaches to a future study.

\subsection*{Acknowledgements}

A.P. thanks Adam Bouland for many helpful discussions. 
Y.S., S.C., D.H., N.K., S.H.S, and M.P would like to thank Ruslan Shaydulin for carefully checking the manuscript and providing feedback, and their colleagues at Global Technology Applied Research of JPMorgan Chase for support and helpful discussions.

\bibliographystyle{unsrt}
\bibliography{references}

\section*{}
\subsection*{Disclaimer}
This paper was prepared for informational purposes with contributions from the Global Technology Applied Research center of JPMorgan Chase \& Co. This paper is not a product of the Research Department of JPMorgan Chase \& Co. or its affiliates. Neither JPMorgan Chase \& Co. nor any of its affiliates makes any explicit or implied representation or warranty and none of them accept any liability in connection with this paper, including, without limitation, with respect to the completeness, accuracy, or reliability of the information contained herein and the potential legal, compliance, tax, or accounting effects thereof. This document is not intended as investment research or investment advice, or as a recommendation, offer, or solicitation for the purchase or sale of any security, financial instrument, financial product or service, or to be used in any way for evaluating the merits of participating in any transaction.

\end{document}